%% file: malleable-scheduling-gs-arxiv.tex
\documentclass[runningheads,oribibl,envcountsame]{llncs}

\usepackage{xcolor}

\usepackage[]{algorithm2e}

\usepackage{amsmath,amssymb}

\usepackage[numbers,sort&compress]{natbib}

\usepackage{hyperref}

\usepackage[capitalize]{cleveref}
\crefformat{footnote}{#2\footnotemark[#1]#3}
\Crefname{appsec}{Appendix}{Appendices}

\hypersetup{
    colorlinks=true,
    linkcolor= blue,
    filecolor=blue,
    citecolor = blue,      
    urlcolor=cyan,
    }

\usepackage[capitalize]{cleveref}

\usepackage{thm-restate}

\newcommand{\LPgs}{\textup{GS}}
\newcommand{\LPgsD}{\textup{GS-D}}
\newcommand{\LPgsa}{\textup{GS-A}} 
\newcommand{\LPwf}{\textup{GS-W}} 
\newcommand{\LPpm}{\textup{GS-P}} 

\newcommand{\maximize}{\textbf{maximize:}}
\newcommand{\minimize}{\textbf{minimize:}}
\newcommand{\subjectto}{\textbf{s.t.:}}

\newcommand{\st}{ : }

\def\mathrlap{\mathpalette\mathrlapinternal} 
\def\mathrlapinternal#1#2{\rlap{$\mathsurround=0pt#1{#2}$}}

\newcommand{\elsum}[1]{\sum_{\mathrlap{#1}}\;}

\newcommand{\argmax}{\operatorname{argmax}}

\title{A Constant-Factor Approximation for Generalized Malleable Scheduling under $M^\natural$-Concave Processing Speeds}

\titlerunning{Generalized Malleable Scheduling under $M^\natural$-Concave Processing Speeds}

\author{%
	Dimitris Fotakis\inst{1} 
	\and Jannik Matuschke\inst{2} 
	\and Orestis Papadigenopoulos\inst{3}
}%

\authorrunning{D. Fotakis, J. Matuschke, and O. Papadigenopoulos}

\institute{%
	National Technical University of Athens, \texttt{fotakis@cs.ntua.gr} 
	\and 
	KU Leuven, \texttt{jannik.matuschke@kuleuven.be}
	\and 
	The University of Texas at Austin, \texttt{papadig@cs.utexas.edu}
}%

\begin{document}

\maketitle

\begin{abstract}
In generalized malleable scheduling, jobs can be allocated and processed simultaneously on multiple machines so as to reduce the overall makespan of the schedule.
The required processing time for each job is determined by the joint processing speed of the allocated machines.
We study the case that processing speeds are job-dependent $M^\natural$-concave functions and provide a constant-factor approximation for this setting, significantly expanding the realm of functions for which such an approximation is possible.
Further, we explore the connection between malleable scheduling and the problem of fairly allocating items to a set of agents with distinct utility functions, devising a black-box reduction that allows to obtain resource-augmented approximation algorithms for the latter.
\end{abstract}

\section{Introduction}

Parallel execution of a job on multiple machines is often used to optimize the overall makespan in time-critical task scheduling systems. Practical applications are numerous and diverse, varying from task scheduling in production and logistics, such as quay crane allocation in naval logistics~\citep{imai2008simultaneous,blazewicz2011berth} and cleaning activities on trains~\citep{bartolini2017scheduling}, to optimizing the performance of computationally demanding tasks, such as web search index update~\citep{wu2015algorithms} and training neural networks~\citep{fujiwara2018effectiveness} (see also \cite{fotakis2020malleable,fotakis2021assigning} for further references and examples). 

The model of \emph{malleable} (a.k.a. \emph{moldable}) jobs, introduced by \citet{du1989complexity}, captures the algorithmic aspects of scheduling jobs that can be executed simultaneously on multiple machines. A malleable job can be assigned to an arbitrary subset of machines to be processed \emph{non-preemptively} and in \emph{unison}, i.e., with the same starting and completion time on each of the allocated machines. Importantly, the scheduler decides on the degree of parallelization for each job, by choosing the set of machines allocated to each job (in contrast to non-malleable parallel machine models, where a single machine is allocated to each task).

Despite the significant interest in the model, most of the work on scheduling malleable jobs considers the case of identical machines, where the processing time of a job only depends on the number of allocated machines. A common assumption is that a job's processing time is non-increasing in the number $k$ of allocated machines, while a job's work (i.e., $k$ times the job's processing time on $k$ machines) is non-decreasing in $k$.  This is usually referred to as the \emph{monotone work assumption} and accounts for communication and coordination overhead due to parallelization. The approximability of makespan minimization in the setting of malleable job scheduling on identical machines is very well understood. Constant-factor approximation algorithms are known since the work of \citet{turek1992approximate}. Following a line of successive improvements in the approximation factor~\citep{mounie1999efficient,jansen2002linear,mounie2007frac32}, two recent results by \citet{jansen2010approximation} and \citet{jansen2018scheduling} implied a polynomial-time approximation scheme for malleable scheduling on identical machines. 

On the other hand, scheduling malleable jobs on non-identical machines has received much less attention. As a natural first step, building on previous work by \citet{correa2015strong} on the closely related \emph{splittable job} model, \citet{fotakis2020malleable} introduced the setting of \emph{speed-implementable processing-time} functions, where each machine $i$ has an unrelated ``speed'' $s_{ij}$ for each job $j$ and a job's processing time is a non-decreasing function of the total allocated speed fulfilling a natural generalization of the non-decreasing work assumption. They devised an LP-based $3.16$-approximation for this setting.

However, as recently observed in \cite{fotakis2021assigning}, the aforementioned models, in which the processing power of a heterogeneous set of machines is expressed by a single scalar, cannot capture the (possibly complicated) combinatorial interaction effects arising among different machines processing the same job. Practical settings where such complicated interdependencies among machines may arise include modern heterogeneous parallel computing systems, typically consisting of CPUs, GPUs, and I/O nodes~\citep{bleuse2017scheduling}, and highly distributed processing systems, where massive parallelization is subject to constraints imposed by the underlying communication network~\citep{bampis2020scheduling}; see \cite{fotakis2021assigning} for further references and examples. Having such practical settings in mind, \citet{fotakis2021assigning} introduced a \emph{generalized malleable scheduling} model, where the processing time $f_j(S) = 1/g_j(S)$ of a job $j$ depends on a job-specific processing speed function $g_j(S)$ of the set of machines $S$ allocated to $j$. In addition to motivating and introducing the model, they derived an LP-based $5$-approximation for scaled matroid rank processing speeds, and a $O(\log \min\{ n, m \})$-approximation algorithm for submodular processing speeds, where $n$ is the number of jobs and $m$ is the number of machines. 

\citet{fotakis2021assigning} left open whether there are processing speed functions more general than scaled matroid rank functions for which generalized malleable scheduling can be approximated within a constant factor. In this work, employing notions and techniques from the field of Discrete Convexity \cite{murota2003discrete}, we present a constant-factor approximation algorithm for job-dependent \emph{$M^\natural$-concave} (a.k.a. \emph{gross substitute}) processing speed functions, thus significantly expanding the realm of functions for which such an approximation is possible. 
We further point out a connection between malleable scheduling and the so-called max-min fair allocation problem, devising a black-box reduction that allows to obtain resource-augmented approximation algorithms for the latter.

\vfill

\subsection{Generalized Malleable Scheduling and Main Results}

To discuss our contribution in more detail, we need to formally introduce the \emph{Generalized Malleable Scheduling} problem. We are given a set of \emph{jobs} $J$ to be assigned to a set of \emph{machines} $M$. Each job $j \in J$ is equipped with a \emph{processing time function} $f_j : 2^M \rightarrow \mathbb{R}_{\geq 0}$ that specifies the time $f_j(S)$ needed for the completion of $j$, when assigned to a subset of machines $S \subseteq M$. 
We assume that functions $f_j$ are accessed through a \emph{value oracle} that, given $S \subseteq M$, returns the value of $f_j(S)$.
A \emph{schedule} consists of two parts: (i) an \emph{assignment} $\mathbf{S} = (S_j)_{j \in J}$ of each job $j \in J$ to a non-empty set of machines $S_j \subseteq M$; and (ii) a \emph{starting time} vector $\mathbf{t} = (t_j)_{j \in J}$, specifying the time $t_j$ at which jobs in $S_j$ start to jointly process job $j$. A schedule is \emph{feasible}, if $S_j \cap S_{j'} = \emptyset$ for all $j, j' \in J$ with $t_j < t_{j'} < t_{j} + f_j(S_j)$, i.e., while a machine is involved in processing a job $j$, it cannot start processing any other job $j'$. The objective is to find a feasible schedule of minimum \emph{makespan} $C(\mathbf{S}, \mathbf{t}) := \max_{j \in J} \{ t_j + f_j(S_j) \}$.

An interesting relaxation of the above \textsc{Scheduling} problem is the \textsc{Assignment} problem, asking for an assignment $\mathbf{S}$ that minimizes the \emph{load} $L(\mathbf{S}) := \max_{i \in M} \sum_{j \in J: i \in S_j} f_j(S_j)$.
Clearly, the load of an assignment is a lower bound on the makespan of any feasible schedule using that same assignment. 

The \emph{processing speed} of a set of machines $S$ for a job $j$ is $g_j(S) := 1/f_j(S)$. Under the assumption that for each job $j \in J$, the processing speed function $g_j$
is submodular, 
\cite[Theorem~1]{fotakis2021assigning} shows that any assignment of maximum machine load $C$ can be transformed in polynomial time into a so-called \emph{well-structured} schedule, where each machine shares at most one job with another machine, of makespan at most $\frac{2e}{e-1}C$. Thus, suffering a small constant-factor loss in the approximation ratio, we can approximate the optimal makespan in Generalized Malleable Scheduling by approximating the \textsc{Assignment} problem. 

Our main contribution is an $O(1)$-approximation for the \textsc{Assignment} problem with $M^\natural$-concave processing speeds (see \cref{sec:submodularity} for definitions of submodularity and $M^\natural$-concavity). Because all $M^\natural$-concave functions are submodular, our results, together with the aforementioned transformation~\cite{fotakis2021assigning}, imply the following theorem. 

\begin{theorem}\label{thm:main}
    There is a polynomial time constant-factor approximation algorithm for \textsc{Assignment} and \textsc{Scheduling} with $M^\natural$-concave processing speeds.
\end{theorem}

\subsection{Submodularity, $M^\natural$-concavity and Matroids}
\label{sec:submodularity}
Let $g : 2^M \rightarrow \mathbb{R}$ be a function. Define the \emph{demand set} for $p \in \mathbb{R}^M$ by 
$\mathcal{D}(g, p) := \argmax_{S \subseteq M} g(S) - \sum_{i \in S} p_i$.
We say that $g$ is 
\begin{itemize}
    \item \emph{submodular} if $g(S \cup T) + g(S \cap T) \leq g(S) + g(T)$ for all $S, T \subseteq M$.
    \item \emph{$M^\natural$-concave} if
for any $p', p'' \in \mathbb{R}^M$ with $p' \leq p''$ and any $S' \in \mathcal{D}(g, p')$ there is an $S'' \in \mathcal{D}(g, p'')$ with $S' \cap \{ i \in M : p'_i = p''_i \} \subseteq S''$. 
\end{itemize}
Submodularity is equivalent to the following diminishing returns property: $g(S \cup \{i\}) - g(S) \geq g(T \cup \{i\}) - g(T)$ for all $S \subseteq T \subseteq M$ and $i \in M \setminus T$.
A \emph{matroid} on the ground set $M$ is a non-empty set family $\mathcal{F} \subseteq 2^M$ such that (i) $T \in \mathcal{F}$ implies $S \in \mathcal{F}$ for all $S \subseteq T$ and (ii) for every $S, T \in \mathcal{F}$ with $|S| < |T|$ there is an $i \in T \setminus S$ such that $S \cup \{i\} \in \mathcal{F}$. It is well known that the the rank function $r(S) := \max_{S \subseteq T : S \in \mathcal{F}} |S|$ for $S \subseteq M$ of any matroid $\mathcal{F}$ is submodular.

$M^\natural$-concavity, also known as the \emph{gross substitutability}, defines an important subclass of submodular functions.
Gross-substitute functions have been widely studied, receiving particular attention in economics and operations research, due to their applications in diverse fields such 
as labor and housing markets~\citep{kelso1982job,gul1999walrasian},
inventory management~\citep{chen2021m}, or structural analysis for engineering systems~\citep{murota2009matrices};
see the survey by \citet{leme2017gross} and the textbook by \citet{murota2003discrete} for an overview of the rich collection of algorithmic results for these class of functions.

In the following, we provide an example for a subclass of $M^\natural$-concave functions, known as  \emph{matroid-based valuations}~\cite{leme2017gross}, which naturally arise in the context of malleable scheduling.

\medskip\noindent{\bf Example: Matroid-Based Valuations.}
Assume that each $j \in J$ is equipped with a set of \emph{processing slots} $V_j$ together with a matroid $\mathcal{F}_j$ on $V_j$. Each slot represents a role that a machine can take to speed up the completion of job $j$. For each a machine $i \in M$ and each slot $v \in V_j$, a weight $w_{iv} \geq 0$ specifies how much $i$ would contribute to the processing of job $j$ when assigned to $v$. When $j$ is processed by a set of machines $S$, each machine in $S$ can fill at most one of $j$'s processing slots and each slot can be taken by at most one machine, i.e., the machines in $S$ are matched to the slots in $V_j$.
The matroid $\mathcal{F}_j$ on $V_j$ specifies which of its slots can be used simultaneously (e.g., the slots could be partitioned into groups and from each group only a limited number of slots may be used).
Thus, a feasible matching of the machines in $S$ to the slots in $V_j$ is a function $\pi : S \rightarrow V_j \cup \{\emptyset\}$ (where $\emptyset$ denotes the machine not being used) such that $\pi(i) = \pi(i') \neq \emptyset$ implies $i = i'$ and such that $\{\pi(i) : i \in S, \pi(i) \neq \emptyset\} \in \mathcal{F}_j$.
Denoting the set of feasible matchings from $S$ to $V_j$ by $\Pi_{S,j}$, the processing speed of a set of machines $S$ for job $j$ is given by $g_j(S) = \max_{\pi \in \Pi_{S, j}} \sum_{i \in S, \pi(i) \neq \emptyset} w_{i\pi(i)}$.

\subsection{Organization of the Paper}
In the following \cref{sec:lp,sec:overview,sec:analysis} we discuss our main result in detail. Our constant-factor approximation is based on a linear programming relaxation of the \textsc{Assignment} problem that we discuss in \cref{sec:lp}, where we show that a pair of optimal primal and dual solution to the LP generate a weighted matroid that can be used to approximate the processing speed functions. This is exploited in the remaining three steps of the algorithm, that, based on the LP solution, partition the set of jobs into three groups, each of which is assigned separately.
We first give an overview of these steps in \cref{sec:overview} and then discuss them in detail in \cref{sec:analysis}.
In \cref{sec:fair-allocation}, we finally turn our attention to the \textsc{Max-min Fair Allocation} (MMFA) problem and present a blackbox reduction that turns approximation algorithms for generalized malleable scheduling into resource-augmented approximation algorithms for MMFA. 
All missing proofs can be found in the appendix.

\section{The Configuration LP}
\label{sec:lp}
In the following we introduce a \emph{configuration LP}, which features a fractional variable $x(S,j) \geq 0$ for each non-empty $S \subseteq M$ and each $j \in J$ (and an auxiliary slack variable $s_i \geq 0$ for each machine $i \in M$).
We will show that this LP, for a given target bound $C$, is feasible if there is an assignment of load at most $C$.

\begin{alignat}{3} 
\maximize \ &&  \sum_{i \in M} s_i  &\tag{\LPgs} \label{LPgs}\\
\subjectto \ && \sum_{S \subseteq M : S \neq \emptyset} \Big(2 - \tfrac{1}{Cg_j(S)}\Big) x(S, j) & \ \geq \ 1 & \quad \forall j \in J \notag\\
&& \sum_{j \in J} \sum_{S \subseteq M : i \in S} \tfrac{1}{g_j(S)} x(S, j) + s_i & \ \leq \ C & \quad \forall i \in M \notag\\
&& x, s & \ \geq \ 0 \notag
\end{alignat}

The dual of \eqref{LPgs} is the following LP.

\begin{alignat*}{3}
\minimize \ && - \sum_{j \in J} \lambda_j +  C \sum_{i \in M} \mu_i & &\tag{\LPgsD} \label{LPgsD}\\
\subjectto \ && \big(2g_j(S) - \tfrac{1}{C}\big) \lambda_j  - \sum_{i \in S} \mu_i & \ \leq \ 0 & \quad \forall j \in J, S \subseteq M, S \neq \emptyset \notag\\
&& \mu_i & \ \geq \ {1} & \forall i \in M \notag\\
&& \lambda_j & \ \geq \ {0} & \forall j \in J \notag
\end{alignat*}

The following lemma reveals that \eqref{LPgs} is indeed a relaxation of \textsc{Assignment}.

\begin{restatable}{lemma}{restatelpfeasibility}
\label{lem:LP-gs-feasibility}
If there exists an assignment $\mathbf{S}$ with $L(\mathbf{S}) \leq C$, then {\em \eqref{LPgs}} is feasible. In addition, if $\lambda, \mu$ is an optimal solution to {\em \eqref{LPgsD}}, then $\lambda_j > 0$ for all $j \in J$.
\end{restatable}

Henceforth, let $(x, s)$ and $(\lambda, \mu)$ be primal-dual optimal solutions to \eqref{LPgs} and \eqref{LPgsD}, respectively. For $j \in J$, define
\begin{align*}
g_j^*(S) := 2g_j(S) - \sum_{i \in S} \frac{\mu_i}{\lambda_j} \quad \text{ and } \quad \mathcal{D}_j := \operatorname{argmax}_{S \subseteq M} g_j^*(S).
\end{align*}
Note that the first constraint of \eqref{LPgsD} is equivalent to 
 $g^*_j(S) \leq 1/C$ for all $S \subseteq M$ and $j \in J$.
Thus complementary slackness implies that sets in the support of $x$ are maximizers of $g^*_j$, as formalized in the next lemma.

\begin{restatable}{lemma}{reatateLPCompSlack}
\label{lem:LP-gs-comp-slack}
  If $x(S, j) > 0$ for $S \subseteq M$ and $j \in J$, then $g_j(S) \geq \frac{1}{2C}$ and $S \in \mathcal{D}_j$.
\end{restatable}

We now observe that the sets $\mathcal{D}_j$ defined in the preceding lemma induce a matroid. This is a consequence of $M^\natural$-concavity, which $g^*_j$ inherits from $g_j$.

\begin{restatable}{lemma}{restateMatroidStructure}
\label{lem:matroidstructure}
  For any $j \in J$, the system \mbox{$\mathcal{F}_j := \{ S \subseteq M \st S \subseteq T \text{ for some } T \in \mathcal{D}_j\}$} is a matroid.
\end{restatable}

Our final lemma in the analysis of \eqref{LPgs} shows that the values $\mu_i/\lambda_j$ provide an approximation for $g_j$ on the matroid $\mathcal{F}_j$. 

\begin{lemma}\label{lem:prices-to-speed}
  Let $S \in \mathcal{F}_j$ and $j \in J$. Then $g_j(S) \geq \frac{1}{2} \sum_{i \in S} \frac{\mu_i}{\lambda_j}$.
\end{lemma}
\begin{proof}
Because $S \in \mathcal{F}_j$ there is $T \in \mathcal{D}_j$ with $S \subseteq T$. Let $i_1, \dots, i_{\ell}$ be an arbitrary ordering of the elements in $S$ and let $S_k = \{i_1, \dots, i_k\}$ for $k \in [\ell]$. We obtain
  \begin{align*}
     2g_j(S) & = \sum_{k = 1}^{\ell} 2g_j(S_k) - 2g_j(S_k \setminus \{i_k\})  \geq \sum_{k = 1}^{\ell} 2g_j(T) - 2g_j(T \setminus \{i_k\})  \geq \sum_{k = 1}^{\ell} \frac{\mu_{i_k}}{\lambda_j},
  \end{align*}
  where the first inequality follows from submodularity and the second follows from the fact that $T \in \mathcal{D}_j$ implies $2g_j(T) - \sum_{i \in T} \frac{\mu_i}{\lambda_j} \geq 2g_j(T \setminus \{i'\}) - \sum_{i \in T \setminus \{i'\}} \frac{\mu_i}{\lambda_j}$ for all $i' \in T$, as $T$ is a maximizer of $g^*_j$.
\qed 
\end{proof}

\section{Overview of the Algorithm}
\label{sec:overview}

Given a target makespan $C$, our algorithm starts from computing an optimal primal-dual solution $(x,s)$ and $(\lambda, \mu)$ to \eqref{LPgs} and \eqref{LPgsD}. Note that such solutions can be computed via dual separation, as the separation problem for \eqref{LPgsD} can be solved via a greedy algorithm for $M^\natural$-concave functions.

If \eqref{LPgs} turns out to be infeasible, we conclude that there is no assignment of maximum load $C$ by \cref{lem:LP-gs-feasibility}. Otherwise, we continue by partitioning the job set into three types. For each of these types, we will show independently how to  turn the corresponding part of the solution to \eqref{LPgs} into an assignment whose load can be bounded by a constant factor times $C$. Binary search for the smallest $C$ for which \eqref{LPgs} is feasible then yields a constant-factor approximation. 

The first type are the jobs that are assigned by our algorithm to a single-machine. For $j \in J$ define $M^+_j := \{i \in M \st g_j(\{i\}) \geq \frac{1}{16C} \}$.
Let $$J_1 := \Big\{\textstyle j \in J \st \sum_{i \in M^+_j} \sum_{S \subseteq M : i \in S} \frac{g_j(\{i\})}{g_j(S)} x(S, j) \geq \frac{1}{16}\Big\}.$$
In \cref{sec:step1}, we show how to obtain an assignment for the jobs in $J_1$ using the LP rounding algorithm of \citet{LST90} for non-malleable unrelated machine scheduling. This establishes the following result.
\begin{lemma}\label{lem:gs-single-machine:result}
Step~1 of the algorithm computes in polynomial time an assignment for the jobs in $J_1$ with a maximum machine load of at most $32 C$.
\end{lemma}

The second type of jobs are those that are assigned predominantly to groups of machines with a relatively low total speed. Formally, define $$\mathcal{S}^2_j := \big\{S \subseteq M \st \textstyle \sum_{i \in S} \mu_i / \lambda_j \leq \frac{4}{C} \big\} \text{ and } J_2 := \big\{j \in J \setminus J_1 \st \sum_{S \in \mathcal{S}^2_j} x(S, j) \geq \frac{1}{8} \big\}.$$ In \cref{sec:step2}, we show how to assign the jobs in $J_2$ via the so-called \textsc{Welfare} problem, which can be solved optimally in polynomial time for gross-substitute functions. This establishes the following result.
\begin{lemma}\label{lem:gs-low-speeds:result}
Step~2 of the algorithm computes in polynomial time an assignment for the jobs in $J_2$ with a maximum machine load of at most $40 C$.
\end{lemma}

Finally, we consider the jobs in $J_3 := J \setminus (J_1 \cup J_2)$. 
Assigning these jobs is more involved than in the preceding cases.

In \cref{sec:step3A}, we modify the fractional LP solution for the jobs in $J_3$ in such a way that the sum of fractional assignments for each machine is bounded by a constant. This transformation uses the fact that we already filtered out jobs using predominantly fast machines or slow assignments in Steps~1 and~2.

In \cref{sec:step3B}, we partition the set of machines for each job according to their weights $\mu_i/\lambda_j$ into classes whose weight differs by a factor of $2$. 
We then use the approximation of the functions $g_j$ via matroids described in \cref{lem:prices-to-speed} to reformulate the problem of assigning a sufficient number of machines from each weight class to each job in $J_3$, while assigning to each machine only a constant number of jobs, as an intersection of two polymatroids. 
The transformed LP solution derived before guarantees the existence of a feasible solution to this polymatroid intersection, from which we can derive an assignment of the jobs in $J_3$, establishing the following result.

\begin{lemma}\label{lem:gs-high-speeds:result}
Step~3 of the algorithm computes in polynomial time an assignment for the jobs in $J_3$ with a maximum machine load of at most $121 C$.
\end{lemma}

By concatenating the assignments for the individual job types we obtain the following result, which implies \cref{thm:main}. 

\begin{theorem}\label{thm:assignment-approx}
There exists a polynomial-time $193$-approximation algorithm for \textsc{Assignment} when processing speeds are $M^\natural$-concave.
\end{theorem}

\section{Description and Analysis of the Intermediate Steps}
\label{sec:analysis}

\subsection{Step 1: Single-machine Assignments for $J_1$} \label{sec:step1}
Recalling the definitions of $M^+_j = \big\{i \in M \st g_j(\{i\}) \geq \frac{1}{16C} \big\}$  for $j \in J$ and of $J_1 = \big\{j \in J \st \sum_{i \in M^+_j} \sum_{S \subseteq M : i \in S} \frac{g_j(\{i\})}{g_j(S)} x(S, j) \geq \frac{1}{16}\big\}$,
consider the following assignment LP:
  \begin{alignat}{3}
    && \textstyle \sum_{i \in M^+_j} y_{ij} & \ \geq \ 1 & \quad \forall\, j \in J_1 \tag{\LPgsa} \label{LPgsa} \\
    && \textstyle \sum_{j \in J_1} \frac{1}{g_j(\{i\})} y_{ij} & \ \leq \ 16C & \quad \forall\, i \in M \notag\\
    && y_{ij} & \ = \ 0 & \quad \forall\, j \in J, i \in M \setminus M^+_j \notag\\
    && y & \ \geq \ 0 \notag
  \end{alignat}
  Note that \eqref{LPgsa} corresponds to an instance of the classic makespan minimization problem on unrelated machines.
  It can be shown that the solution $x$ to \eqref{LPgs} induces a feasible solution to  \eqref{LPgsa}. 
  By applying the rounding algorithm of \citet{LST90} to an extreme point solution of \eqref{LPgsa}, we get an assignment of the jobs in $J_1$ in which each machine $i \in M$ receives a load of at most $16C + 1/g_j(\{i\})$ for some $j \in J_1$ with $i \in M^+_j$.
  Because $g_j(\{i\}) \geq \frac{1}{16 C}$ for $j \in J_1$ and $i \in M^+_j$, the load of this assignment is at most $32C$ and \cref{lem:gs-single-machine:result} follows.

\subsection{Step 2: Assignments with Low Total Speed for $J_2$} \label{sec:step2}

Recalling the definitions of $\mathcal{S}^2_j = \{S \subseteq M \st \textstyle \sum_{i \in S} \mu_i / \lambda_j \leq 4 / C\}$ and of $J_2 = \{j \in J \setminus J_1 \st \sum_{S \in \mathcal{S}^2_j} x(S, j) \geq \frac{1}{8}\}$ consider the following LP:
\begin{alignat}{3}
\maximize \quad && \sum_{S \subseteq M} \sum_{j \in J_2}   g^*_j(S) &z(S, j) \tag{\LPwf} \label{LPwf}\\
\subjectto \quad && \sum_{S \subseteq M} z(S, j) & \ \leq \ 1 & \quad \forall\, j \in J_2 \notag\\
&& \sum_{j \in J_2} \sum_{S \subseteq M : i \in S} z(S, j) & \ \leq \ 20 & \quad \forall\, i \in M \notag \\
&& z & \ \geq \ 0 \notag 
\end{alignat}
Again, it can be shown that the solution $x$ to \eqref{LPgs} induces a feasible solution to \eqref{LPwf}.
Moreover, \eqref{LPwf} corresponds to an LP relaxation of the so-called \textsc{Welfare Maximization} problem. 
Because the functions $g^*_j$ are $M^\natural$-concave, this LP is totally dual integral, and integer optimal solutions can be found in polynomial time \cite{leme2020computing}. 
Thus, let $z$ be such an integer optimal solution.
\cref{lem:gs-welfare} below guarantees that for each $j \in J_2$ there is a set $S_j \subseteq M$ with $z(S_j, j) = 1$ and $g_j(S_j) \geq 1 / 2C$. Since each machine participates in the execution of at most $20$ jobs, each of processing time $2C$, we obtain an assignment with load at most $40C$ for the jobs in $J_2$, and \cref{lem:gs-low-speeds:result} follows.

\begin{restatable}{lemma}{restateGSWelfare}
\label{lem:gs-welfare}
  Let $z$ be an integral optimal to \eqref{LPwf}. Then for each $j \in J_2$ there is a set $S_j$ with $z(S_j, j) = 1$ and $g_j(S_j) \geq \frac{1}{2C}$.
\end{restatable}
\begin{proof}[sketch]
Using the construction of $J_2$, it can be shown that setting $z'(S, j) := x(S, j) / \sum_{S' \in \mathcal{S}^2_j} x(S', j)$ for $j \in J_2$ and $S \in \mathcal{S}^2_j$ yields a solution of value $\sum_{j \in J_2} \sum_{S \subseteq M} g_j^*(S) z'(S, j) = |J_2| / C$ to \eqref{LPwf}.
Now let $z$ be an integral optimal solution.
Note that $\sum_{j \in J_2} \sum_{S \subseteq M} g^*(S) z(S, j) \geq |J_2| / C$ by total dual integrality. Because $g_j^*(S) \leq \frac{1}{C}$ for $S \subseteq M$ and $j \in J$ by the constraints of \eqref{LPgsD}, we conclude that for each $j \in J_2$ there is a set $S_j$ with $z(S_j, j) = 1$ and $g_j^*(S_j) = \frac{1}{C}$. The latter implies $g_j(S_j) = \big(g_j^*(S_j) + \sum_{i \in S_j} \mu_i/\lambda_j\big)/2 \geq \frac{1}{2C}$.\qed
\end{proof}

\subsection{Step 3A: Splitting Assignments with High Total Speed}
\label{sec:step3A}

Recall that $J_3 = J \setminus (J_1 \cup J_2)$.
For $j \in J_3$ define 
$$\textstyle \mathcal{S}^3_j := \left\{ S \subseteq M \st x(S, j) > 0, \; S \notin \mathcal{S}^2_j, \; \sum_{i \in S \setminus M^+_j} \frac{\mu_i}{\lambda_j} > \sum_{i \in S \cap M^+_j} \frac{\mu_i}{\lambda_j} \right\}.$$
We construct a new fractional assignment $x' : 2^M \times J_3 \rightarrow \mathbb{R}_{\geq 0}$ as follows. For each $j \in J_3$ and each $T \in \mathcal{S}^3_j$, find a partition $\mathcal{A}_j(T)$ of $T \setminus M^+_j$ such that
\begin{itemize}
\item $\sum_{i \in A} \mu_i / \lambda_j \leq 4/C$ for all $A \in \mathcal{A}_j(T)$.
\item $\sum_{i \in A} \mu_i / \lambda_j + \sum_{i \in A'} \mu_i / \lambda_j > 4/C$ for all distinct sets $A, A' \in \mathcal{A}_j(T)$.
\end{itemize}
Because $\mu_i/\lambda_j \leq 4/C$ for all $i \in M \setminus M^+_j$ such a partition exists and can be constructed greedily.
For $j \in J_3$ and $S \subseteq M$ let $\mathcal{T}_j(S) := \{T \in \mathcal{S}^3_j \st S \in \mathcal{A}_j(T)\}$ and define
 $$\bar{x}(S, j) := \sum_{T \in \mathcal{T}_j(S)} \frac{g_j(S)}{\sum_{A \in \mathcal{A}_j(T)} g_j(A)} x(T, j) \quad \text{ and } \quad x'(S, j) := \frac{\bar{x}(S, j)}{\gamma_j},$$
where $\gamma_j := \sum_{S \subseteq M} \bar{x}(S, j)$.
For ease of notation, define 
$x'_{ij} := \sum_{S \subseteq M \st i \in S} x'(S,j)$.

The modified fractional assignment $x'$ exhibits several properties that will be useful in the construction of an integral assignment for $J_3$.

\begin{restatable}{lemma}{restategssplit}
\label{lem:gs-split-solution}
  The assignment $x'$ described above fulfills the following properties:\\[-15pt]
  \begin{enumerate}
    \item $\sum_{S \subseteq M} x'(S, j) = 1$ for all $j \in J_3$,\label{eq:xprime-total}
    \item $S \cap M^+_j = \emptyset$ and $S \in \mathcal{F}_j$ for all $S \subseteq M$ and $j \in J_3$ with $x'(S, j) > 0$.\label{eq:xprime-matroid}
    \item $\sum_{j \in J_3} x'_{ij} \leq 26$ for all $i \in M$,\label{eq:xprime-load}
    \item $\sum_{i \in M} \frac{\mu_i}{\lambda_j} x'_{ij} \geq \frac{79}{40C}$ for all $j \in J_3$,\label{eq:xprime-prices}
  \end{enumerate}
\end{restatable}

\begin{proof}[sketch]
The first two properties of the lemma follow directly from construction of $x'$ and from the fact that $S \in \mathcal{F}_j$ when $x(S, j) > 0$.
Intuitively, the construction of $x'$ splits assignments of high total weight into assignments of moderate weight, then scaling the solution by $1/\gamma_j$ so as to compensate for the fact that machines in $M^+_j$ and assignments in $\mathcal{S}^2_j$ are ignored.
The main part of the proof is to show that this scaling factor and thus the blow-up in makespan can be bounded by a small constant.
Once this is established, Property~\ref{eq:xprime-load} follows from the fact that $x'(S, j) > 0$ implies $g_j(S) \leq \frac{1}{2} \big(\frac{1}{C} + \sum_{i \in S} \frac{\mu_i}{\lambda_j}\big) \leq \frac{5}{2C}$ by construction and hence the total sum of fractional assignments for every machine is bounded by a constant.
Finally, Property~\ref{eq:xprime-prices}  can be derived from the fact that $\sum_{i \in S \cap M^+_j} \frac{\mu_i}{\lambda_j} > \frac{1}{2} \sum_{i \in S} \frac{\mu_i}{\lambda_j} \geq \frac{2}{C}$ for all $S \in \mathcal{S}^3_j$. \qed
\end{proof}

\subsection{Step 3B: Constructing an Integer Assignment for $J_3$}
\label{sec:step3B}
For $j \in J_3$ and $k \in \mathbb{Z}$ define 
\begin{align*}
   \textstyle M_{jk} := \left\{i \in M \st \frac{1}{2^{k+1} C} < \frac{\mu_i}{\lambda_j} \leq \frac{1}{2^{k}C}\right\} \quad \text{and} \quad d_{jk} := \bigg\lfloor\; \sum_{i \in M_{jk}} x'_{ij} \bigg\rfloor.
\end{align*}
Note that there are only polynomially many $k \in \mathbb{Z}$ with $d_{jk} > 0$. 

Furthermore, define $r_j'(U) := \sum_{k \in \mathbb{Z}}\, \min \, \{r_j(U \cap M_{jk}),\, d_{jk}\}$ for $j \in J_3$ and $U \subseteq M$,
where $r_j$ is the rank function of the matroid $\mathcal{F}_j$.
Consider the LP:
\begin{alignat}{3}
\maximize \quad && \sum_{i \in M} \sum_{j \in J_3}  y_{ij} \tag{\LPpm} \label{LPpm}\\
\subjectto\quad && \sum_{j \in J_3} y_{ij} & \ \leq \ 26 & \quad \forall\, i \in M \notag\\
&& \sum_{i \in U} y_{ij} & \ \leq \ r_j'(U) & \quad \forall\, U \subseteq M,\, j \in J_3 \notag \\
&& y & \ \geq \ 0 \notag
\end{alignat}
Let $y$ be an integer optimal solution to the following LP (such a point exists and can be computed in polynomial time due to \cref{lem:pm-LP-gs,lem:gs-pmi-value} below).
The assignment for $J_3$ is constructed by setting $S_j := \{i \in M \st y_{ij} > 0\}$ for $j \in J_3$.

Because $y$ is integral, $|\{j \in J_3 \st i \in S_j\}| \leq 26$ for each $i \in M$.
Moreover, \cref{lem:gs-processing-time-bound} guarantees that $f_j(S_j) \leq 320C/69$ for each $j \in J_3$.  We have thus found an assignment of the jobs in $J_3$ with maximum load $26 \cdot 320C/69 < 121C$.
This completes the proof of \cref{lem:gs-high-speeds:result} and the description of the algorithm.

To complete the analysis, it remains to prove \cref{lem:pm-LP-gs,lem:gs-pmi-value,lem:gs-processing-time-bound} invoked above.
We first observe that the function $r'_j$ is submodular for each $j \in J_3$ and hence  \eqref{LPpm} is indeed the intersection of two polymatroids. As a consequence, we obtain the following lemma.

\begin{restatable}{lemma}{restatePMIntegrality}
\label{lem:pm-LP-gs}
All extreme points of \eqref{LPpm} are integral. If \eqref{LPpm} is feasible, an optimal extreme point can be computed in polynomial time.
\end{restatable}

We next show that \eqref{LPpm} has a feasible solution that attains the bound $\sum_{i \in M_{jk}} y_{ij} \leq d_{jk}$ implicit in the definition of $r'_j$ for each $k$ and $j$ with equality.

\begin{restatable}{lemma}{restatePolymatroidValue}\label{lem:gs-pmi-value}
\eqref{LPpm} has a feasible solution of value $\sum_{j \in J_3} \sum_{k \in \mathbb{Z}} d_{jk}$.
\end{restatable}
\begin{proof}[sketch]
For $j \in J_3$ and $i \in M_{jk}$, let $y'_{ij} := \frac{d_{jk}}{\sum_{i' \in M_{jk}} x'_{i'j}} x'_{ij}$ if \mbox{$d_{jk} > 0$} and $y'_{ij} := 0$ otherwise.
By construction, $\sum_{i \in M} \sum_{j \in J_3} y'_{ij} = \sum_{j \in J_3} \sum_{k \in \mathbb{Z}} d_{jk}$. 

To see that $y'$ is also feasible, observe that $\sum_{j \in J_3} y'_{ij} \leq \sum_{j \in J_3} x'_{ij} \leq 26$ by Property~\ref{eq:xprime-load} of \cref{lem:gs-split-solution} and the fact that $y'_{ij} \leq x'_{ij}$ by construction.
Moreover, Properties~\ref{eq:xprime-total}~and~\ref{eq:xprime-matroid} imply $\sum_{i \in U \cap M_{jk}} y'_{ij} \leq \min \{d_{jk},\, r(U \cap M_{jk})\}$ for all $U \subseteq M$ and all $k \in \mathbb{Z}$.
As $M_{jk} \cap M_{jk'} = \emptyset$ for $k \neq k'$, we conclude that $\sum_{i \in U} y'_{ij} \leq r'(U)$. Hence $y'$ is a feasible solution to \eqref{LPpm}. 
\qed
\end{proof}

To prove \cref{lem:gs-processing-time-bound}, we use the following consequence of the properties of $x'$ described in \cref{lem:gs-split-solution}.

\begin{restatable}{lemma}{restateSpeedSum}\label{lem:gs-speed-sum}
  $\sum_{k \in \mathbb{Z}} \frac{1}{2^kC} d_{jk} \geq \frac{69}{40C}$ for every $j \in J_3$.
\end{restatable}
\begin{proof}
Note that $x'(S, j) > 0$ implies $S \cap M^+_j = \emptyset$ by Property~\ref{eq:xprime-matroid} of \cref{lem:gs-split-solution} and hence $\frac{\mu_i}{\lambda_j} \leq 2g_j(\{i\}) < \frac{1}{8C}$ for all $i \in S$. Hence $d_{jk} = 0$ for $k < 3$ and thus
  \begin{align*}
    \sum_{k \in \mathbb{Z}} \frac{1}{2^kC} d_{jk} & \; = \; \sum_{k = 3}^{\infty} \frac{1}{2^kC} \bigg\lfloor \sum_{i \in M_{jk}} x'_{ij} \bigg\rfloor
    \; \geq \; \elsum{i \in M} \frac{\mu_i}{\lambda_j} x'_{ij} \; - \sum_{k = 3}^{\infty} \frac{1}{2^kC}\\
    & \; \geq \; \elsum{S \subseteq M}  \sum_{i \in S} \frac{\mu_i}{\lambda_j} x'(S, j) \; - \; \frac{1}{4C}
    \; \geq \; \frac{79}{40C} - \frac{1}{4C} = \frac{69}{40C},
  \end{align*}
  where the last inequality follows from Property~\ref{eq:xprime-prices} of \cref{lem:gs-split-solution}. \qed
\end{proof}

We are now ready to combine \cref{lem:gs-pmi-value,lem:gs-speed-sum} and the fact that \mbox{$S_j \in \mathcal{F}_j$} to show that each job in $J_3$ is indeed assigned sufficient processing speed.

\begin{lemma}\label{lem:gs-processing-time-bound}
  For $j \in J_3$ let $S_j := \{i \in M \st y_{ij} > 0\}$. Then $g_j(S_j) \geq \frac{69}{320C}$.
\end{lemma}
\begin{proof}
Note that
  \begin{align*}
    \sum_{j \in J_3} \sum_{k \in \mathbb{Z}} d_{jk} & \leq \sum_{j \in J_3} \sum_{i \in M} y_{ij} = \sum_{j \in J_3} \sum_{k \in \mathbb{Z}} \sum_{i \in M_{jk}} y_{ij} \leq \sum_{j \in J_3} \sum_{k \in \mathbb{Z}} r'_j(M_{jk}) \leq \sum_{j \in J_3} \sum_{k \in \mathbb{Z}} d_{jk} 
  \end{align*}
  where the first identity follows from by \cref{lem:gs-pmi-value} and the final two inequalities follow from feasibility of $y$ and definition of $r'_j$, respectively. We conclude that all inequalities are fulfilled with equality, which is only possible if $\sum_{i \in M_{jk}} y_{ij} = d_{jk}$ for all $j \in J_3$ and $k \in \mathbb{Z}$. 

  Let $T_j \subseteq S_j$ with $T_j \in \mathcal{F}_j$ maximizing $\sum_{i \in T_j} \frac{\mu_i}{\lambda_j}$ computed by the matroid greedy algorithm.
  The greedy algorithm ensures
  $\left|T_j \cap \bigcup_{k = 0}^{\ell} M_{jk}\right| \geq r_j(S_j \cap M_{j\ell})$ for all $\ell \in \mathbb{Z}$. We conclude that
  \begin{align*}
    \sum_{i \in T_j} \frac{\mu_i}{\lambda_j} & \geq \frac{1}{2} \sum_{k = 0}^{\infty} r_j(S_j \cap M_{jk}) \frac{1}{2^{k+1}C} \geq \frac{1}{2} \sum_{k = 0}^{\infty} \frac{1}{2^{k+1}C} \elsum{i \in M_{jk}} y_{ij} \\
    & \geq \frac{1}{4} \sum_{k = 0}^{\infty} \frac{1}{2^{k}C} d_{jk} \geq \frac{69}{160C},
  \end{align*}
  where the last inequality follows from \cref{lem:gs-speed-sum}.
Because $T_j \in \mathcal{F}_j$, we conclude that $g_j(S_j) \geq g_j(T_j) \geq \frac{69}{320C}$ by \cref{lem:prices-to-speed}.\qed
\end{proof}

\section{Generalized Malleable Jobs and Fair Allocations}
\label{sec:fair-allocation}

In this section, we explore an interesting relation between generalized malleable scheduling and \textsc{Max-min Fair Allocation} (MMFA). In this problem, we are given a set of {\em items} $I$ and a set of {\em agents} $A$. Each agent $j \in A$ has a \emph{utility function} $u_j: 2^{I} \to \mathbb{R}_{\geq 0}$ on the items. Our goal is to assign the items to the agents in a way to maximize the minimum utility, that is, to find an assignment $\mathbf{T}$ with $|\{j \in A : i \in T_j\}| \leq 1$ for all $i \in I$ so as to maximize $\min_{j \in A} u_j(T_j)$.

We show that any approximation algorithm for malleable scheduling implies a \emph{resource-augmented} approximation for MMFA in which some items may be assigned to a small number of agents (this can be interpreted as a moderate way of splitting these items, e.g., multiple agents sharing a resource by taking turns).

To formalize this result, we establish two definitions.
A $\beta$-augmented $\alpha$-approximation algorithm for MMFA is an algorithm that given an MMFA instance computes in polynomial time an assignment $\mathbf{T}$ with $\min_{j \in A} u_j(T_j) \geq \frac{1}{\alpha} V^*$ and $|\{j \in A : i \in T_j\}| \leq \beta$ for all $i \in I$, where $V^*$ is the optimal solution value of the MMFA instance.
Moreover, we say that a class of set functions $\mathcal{C}$ is \emph{closed under truncation} if for any $h \in \mathcal{C}$ and any $t \in \mathbb{R}_{\geq 0}$, the function $h^t$ defined by $h^t(S) := \min \{h(S), t\}$ is contained in $\mathcal{C}$.

\begin{theorem}
	\label{thm:mmfa}
  Let $\mathcal{C}$ be a class of set functions closed under truncation.
  If there is an $\alpha$-approximation algorithm for \textsc{Assignment} with processing speeds from $\mathcal{C}$, then there is an $\lfloor \alpha \rfloor$-augmented $\alpha$-approximation algorithm for MMFA with utilities from $\mathcal{C}$.
\end{theorem}
\begin{proof}
    Given an instance of MMFA and a target value $V$ for the minimum utility (to be determined by binary search),
    let $M := I$ and $J := A$, i.e., we introduce a machine for each item and a job for each agent.
    Define processing speeds $g_j$ for $j \in J$ by $g_j(S) := \min \{u_j(S), V\}$ for $S \subseteq M$.
    Now apply the $\alpha$-approximation algorithm to this \textsc{Assignment} instance and obtain an assignment $\mathbf{S}$.
    
    If $\max_{i \in M} \sum_{j \in J : i \in S_j} 1/g_j(S_j) \leq \alpha / V$, then return the assignment $\mathbf{S}$ as a solution to the MMFA instance (note that in this case, $g_j(S_j) \geq V/\alpha$ for each $j \in J$ and $|\{j \in J : i \in S_j\}| \leq \lfloor \alpha \rfloor$ for all $i \in M$, because $g_j(S_j) \leq V$ for all $j \in J$).
    If, on the other hand, $\max_{i \in M} \sum_{j \in J : i \in S_j} 1/g_j(S_j) > \alpha / V$, then we conclude that the MMFA instance does not allow for a solution of value $V$ (because such a solution would have load $1/V$ in the \textsc{Assignment} instance). \qed
\end{proof}

This black-box reduction, together with the $3.16$-approximation for speed-implementable functions~\cite{fotakis2020malleable} implies a $3$-augmented $3.16$-approxi\-mation for the well-known \textsc{Santa Claus} problem~\citep{bansal2006santa} (the special case of MMFA with linear utilities). A more careful analysis delivers the following stronger result:
\begin{restatable}{corollary}{restateSantaClaus}
\label{cor:santa}
    There is a $2$-augmented $2$-approximation for \textsc{Santa Claus}.
\end{restatable}

Although $M^\natural$-concave functions are not closed under truncation as defined above, a slightly different form of truncating processing speeds allows us to apply the reduction on the constant-factor approximation for $M^\natural$-concave processing speeds presented in this paper. We thus obtain the following result.

\begin{restatable}{corollary}{restateMMFAnatural}
\label{cor:mmfa-natural}
There exists a $\mathcal{O}(1)$-augmentation $\mathcal{O}(1)$-approximation algorithm for MMFA with $M^\natural$-concave utilities. 
\end{restatable}

\section{Conclusion}

In this paper we have presented a constant-factor approximation for generalized malleable scheduling under $M^\natural$-concave processing speeds.
To achieve a constant approximation guarantee, our algorithm makes extensive use of structural results from discrete convex analysis. We think that some of the techniques, such as the rounding technique for weighted polymatroids in \cref{sec:step3B}, might be of independent interest and applicable in other contexts as well.
We have not made any attempt to optimize the constant in the approximation ratio, but we expect that significant additional insights are required to achieve a reasonably small (single-digit) approximation guarantee. 

An intriguing open question is whether there exists a constant-factor approximation for generalized malleable scheduling under \emph{submodular} processing speeds, for which only a logarithmic approximation is known to date, together with a strong inapproximability result for the more general XOS functions~\cite{fotakis2021assigning}. 
Our present work is a significant progress in this direction, but several steps of our algorithm, such as the approximation of processing speeds via weighted matroids and the elimination of low-speed assignments crucially use structural properties of $M^\natural$-concavity not present in submodular functions. Overcoming these issues is an interesting direction for future research.

\bibliographystyle{splncs04nat}
\renewcommand{\bibsection}{\section*{References}}
\bibliography{ref}

\input{appendix.tex}

\end{document}

%% file: appendix.tex
\clearpage

\appendix

\spnewtheorem{nclaim}{Claim}{\bfseries}{\itshape}

\crefalias{section}{appsec}
\crefalias{subsection}{appsec}

\section{Appendix to \cref{sec:lp}: Proofs of \cref{lem:LP-gs-feasibility,lem:LP-gs-comp-slack,lem:matroidstructure}}

\restatelpfeasibility*
\begin{proof}
Let $\mathbf{S}$ be an assignment with $L(\mathbf{S}) \leq C$. 
For each $j \in J$, we set $x'(S,j) = 1$ if $S = S_j$, and $x'(S,j) = 0$ otherwise. 
For each machine $i \in M$, we set $s'_i = C - \sum_{j \in J:i \in S_j} \frac{1}{g_j(S_j)}$. 
Note that
$$\sum_{S \subseteq M : S \neq \emptyset} \left(2 - \frac{1}{Cg_j(S)}\right) x'(S, j) = 2 - \frac{1}{Cg_j(S_j)} \geq 1$$ for any $j \in J$, because $L(\mathbf{S}) \leq C$ implies $g_j(S_j) \geq \frac{1}{C}$ for each each $j \in J$.
Moreover, observe that
$$\sum_{j \in J} \sum_{S \subseteq M:i \in M} \frac{1}{g_j(S)} x'(S,j) + s'_i = \sum_{j \in J: i \in S_j} \frac{1}{g_j(S_j)} + s'_i = C$$
for $i \in M$ by definition of $x'$ and $s'$. Finally, note that $x' \geq 0$ by definition and $s' \geq 0$ because the maximum machine load in $\mathbf{S}$ is at most $C$. 

For the second part of the statement, consider an optimal solution $(\lambda^*, \mu^*)$ to \eqref{LPgsD}~and let $(x^*, s^*)$
be an optimal solution to \eqref{LPgs}. 
Note that feasibility of $x^*$ implies that for any $j \in J$ there is at least one non-empty set $S \subseteq M$ with $x^*(S, j) > 0$ and $2 - \frac{1}{Cg_j(S)} > 0$.
Then complementary slackness implies
$\left(2 g_j(S) - \frac{1}{C}\right) \lambda_j - \sum_{i \in S} \mu_i = 0$, from which we derive
$$\lambda_j = \frac{\sum_{i \in S} \mu_i} { \left(2 g_j(S) - \frac{1}{C}\right) } > 0$$ because $\mu_i \geq 1$ for all $i \in M$ by dual feasibility. 
\qed
\end{proof}

\reatateLPCompSlack*
\begin{proof}
Let $S \subseteq M$ and $j \in J$ with $x(S, j) > 0$. Note that
$$(2 g_j(S) - \frac{1}{C} ) \lambda_j - \sum_{i \in S} \mu_i = 0$$ by complementary slackness.
This implies $2g_j(S) \geq \frac{1}{C}$, because $\lambda_j > 0$ and $\sum_{i \in S} \mu_i \geq 0$. Furthermore, it implies $g^*_j(S) = \frac{1}{C}$ by simple rearrangement.
Because $g^*_j(S') \leq \frac{1}{C}$ for all $S' \subseteq M$, we obtain $S \in \mathcal{D}_j$. \qed
\end{proof}

\restateMatroidStructure*
\begin{proof}
Let $j \in J$.
Because $g_j$ is an $M^\natural$-concave function, $\mathcal{D}_j$ is an $\textup{M}^\natural$-convex set~\citep{leme2017gross}, i.e., for any $S, T \in \mathcal{D}_j$ and $i \in S \setminus T$ one of the following two statements holds:
  \begin{enumerate}
    \item $T \cup \{i\} \in \mathcal{D}_j$ and $S \setminus \{i\} \in \mathcal{D}_j$.
    \item There is $i' \in T \setminus S$ such that $S \setminus \{i\} \cup \{i'\} \in \mathcal{D}_j$ and $T \setminus \{i'\} \cup \{i\} \in \mathcal{D}_j$.\label{eq:m-convex-exchange}
  \end{enumerate}
  Note that the set $\mathcal{F}_j$ is downward closed by construction. 
  Now let $S, T$ be two bases of $\mathcal{F}_j$.
  Note that $S$ and $T$ must be $\subseteq$-maximal sets in $\mathcal{D}_j$.
  Let $i \in S \setminus T$.
  Note that maximility of $T$ implies $T \cup \{i\} \notin \mathcal{D}_j$ and thus there must be $i' \in T \setminus S$ with $S \setminus \{i\} \cup \{i'\} \in \mathcal{D}_j$ and $T \setminus \{i'\} \cup \{i\} \in \mathcal{D}_j$. Hence $\mathcal{F}_j$ fulfills the basis exchange axiom and is a matroid. 
\qed
\end{proof}

\section{Appendix to \cref{sec:overview}: Solving the LP \eqref{LPgs}}
\label{app:solving-gs-lp}

\begin{lemma}
Optimal solutions to \eqref{LPgs} can be computed in polynomial time.
\end{lemma}
\begin{proof}
We construct a polynomial-time separation oracle for the dual \eqref{LPgsD}. Given such an oracle, an optimal dual solution can be computed in polynomial-time using the ellipsoid method. Using the fact that this solution is of polynomial-size, an optimal solution to the primal can recovered by solving \eqref{LPgs} restricted to those variables for which the corresponding dual constraints were generated during the run of the ellipsoid method for the dual (note that this restricted LP is of polynomial size, as the ellipsoid method terminated after a polynomial number of steps). 

We now describe the separation oracle. Given any values $(\lambda, \mu)$, we can trivially verify the feasibility of $\lambda_j, \mu_i \geq 0$ and $\mu_i \leq 1$ for all $i \in M$ and $j \in J$ in linear time. 
In order to verify of the constraints
$(2g_j(S) - \frac{1}{C})\lambda_j - \sum_{i \in S} \mu_i \leq 0$ for all $j \in J$ and $S \subseteq M, S \neq \emptyset$, consider any fixed $j \in J$ separately. 
If $\lambda_j = 0$, we know that the constraint is fulfilled for all $S \subseteq M$.
If $\lambda_j > 0$, let $S^*$ be an optimal solution to $\max_{S \subseteq M}\{g_j(S) - \frac{\mu}{2\lambda_j}\}$
(such a solution can be found in polynomial time using the greedy algorithm because $g_j$ fulfills gross substitutability). 
If $g_j(S^*) - \frac{\mu}{2\lambda_j} > \frac{1}{2C}$, we found a separating hyperplane. Otherwise, we can conclude that $(\lambda, \mu)$ is feasible.
\qed
\end{proof}

\section{Appendix to \cref{sec:analysis}}

\subsection{\cref{sec:step1}: Feasibility of \eqref{LPgsa}}

\begin{lemma}\label{lem:gs-single-machine}
\eqref{LPgsa} has a feasible solution.
\end{lemma}
\begin{proof}
For any $j \in J_1$, we define $y'_{ij} := 16 \sum_{S \subseteq M : i \in S} \frac{g_j(\{i\})}{g_j(S)} x(S, j)$ if $i \in M^+_j$, and $y'_{ij} := 0$, otherwise. Observe that $\sum_{i \in M^+_j} y'_{ij} \geq 1$ for all $j \in J_1$ by definition of $J_1$. Furthermore, note that 
  $$\textstyle \sum_{j \in J_1} \frac{1}{g_j(\{i\})} y'_{ij} \leq 16 \sum_{j \in J_1} \sum_{S \subseteq M : i \in S} \frac{1}{g_j(S)} x(S, j) \leq 16C$$ for all $i \in M$ by definition of $y'$ and by the constraints of \eqref{LPgs}. Thus $y'$ is a feasible solution to \eqref{LPgsa}. 
\qed
\end{proof}

\subsection{\cref{sec:step2}: Proof of \cref{lem:gs-welfare}}

\restateGSWelfare*
\begin{proof}
We first show that \eqref{LPwf} has optimal value $\frac{|J_2|}{C}$. To this end, define $\gamma_j := \sum_{S \in \mathcal{S}^2_j} x(S, j) \geq \frac{1}{8}$ for $j \in J_2$. 
  Let $z'(S, j) := \frac{1}{\gamma_j} x(S, j)$ for $S \in \mathcal{S}^2_j$ and let $z'(S, j) := 0$ for all $S \subseteq M$ with $S \notin \mathcal{S}^2_j$. 
  We show that $z'$ is a feasible solution to \eqref{LPwf} with value $\frac{|J_2|}{C}$.

  First observe that $\sum_{S \subseteq M} z'(S, j) = 1$ for each $j \in J_2$ by construction. Note further that by the constraints of \eqref{LPgsD} and the fact that $\lambda_j > 0$ by Lemma \ref{lem:LP-gs-feasibility}, we have
  $$g_j(S) \leq \frac{1}{2C} + \sum_{i \in S} \frac{\mu_i}{2\lambda_j} \leq \frac{5}{2C}$$ for all $j \in J$ and $S \in \mathcal{S}^2_j$.
  Hence,
  \begin{align*}
    \sum_{j \in J_2} \sum_{S \subseteq M : i \in S} z'(S, j) & \leq \sum_{j \in J_2} \sum_{S \in \mathcal{S}^2_j : i \in S} \frac{5}{2Cg_j(S)} \cdot \frac{x(S, j)}{\gamma_j}\\
   & \leq \frac{20}{C} \sum_{j \in J} \sum_{S \subseteq M} \frac{1}{g_j(S)} x(S, j) \leq 20
  \end{align*}
   for each $i \in M$, where the second inequality follows from $\gamma_j \geq \frac{1}{8}$ for all $j \in J_2$, while the third inequality follows by the primal constraints of \eqref{LPgs}. Thus $z'$ is a feasible solution to $\LPwf$.
  Because $z'(S, j) > 0$ implies $x(S, j) > 0$ and hence $g_j^*(S) = \frac{1}{C}$ by \cref{lem:LP-gs-comp-slack}, we conclude that $\sum_{j \in J_2} \sum_{S \subseteq M} g_j^*(S) z'(S, j) = \frac{|J_2|}{C}$.
  
Now let $z$ be an integral optimal solution to \eqref{LPwf}.
Because \eqref{LPwf} is totally dual integral, the above implies that $\sum_{j \in J_2} \sum_{S \subseteq M} g^*(S) z(S, j) = \frac{|J_2|}{C}$. Because $g_j^*(S) \leq \frac{1}{C}$ for all $S \subseteq M$ and $j \in J$ by the constraints of \eqref{LPgsD}, we conclude that for each $j \in J_2$ there is a unique set $S_j$ with $z(S_j, j) = 1$ and $g_j^*(S_j) = \frac{1}{C}$. Note that the latter implies $g_j(S_j) = \frac{1}{2}\big(g_j^*(S_j) + \sum_{i \in S_j} \frac{\mu_i}{\lambda_j}\big) \geq \frac{1}{2C}$.
\qed
\end{proof}

\subsection{\cref{sec:step3A}: Proof of \cref{lem:gs-split-solution}}
\label{app:proof-split-lemma}

In oder to proof \cref{lem:gs-split-solution}, we first establish the following lower bound on $\gamma_j$ as defined in \cref{sec:step3A}.

\begin{lemma}
\label{lem:split-gamma}
For every $j \in J_3$, it holds that $\gamma_j = \sum_{T \in \mathcal{S}^3_j} x(T, j) \geq \frac{39}{160}$.
\end{lemma}
\begin{proof}
First observe that
  \begin{align*}
    \gamma_j & \; = \; \sum_{S \subseteq M} \bar{x}(S, j) \; = \; \sum_{S \subseteq M} \sum_{T \in \mathcal{T}_j(S)} \frac{g_j(S)}{\sum_{A \in \mathcal{A}_j(T)} g_j(A)}x(T, j) \\
    & \; = \; \sum_{T \in \mathcal{S}^3_j} \frac{\sum_{A \in \mathcal{A}_j(T)} g_j(A)}{\sum_{A \in \mathcal{A}_j(T)} g_j(A)}x(T, j) \; = \; \sum_{T \in \mathcal{S}^3_j} x(T, j).
  \end{align*}
  
  To bound the value of $\sum_{T \in \mathcal{S}^3_j} x(T, j)$ we establish two claims.
  For the first claim, we define $$\textstyle \mathcal{S}^1_j := \left\{S \subseteq M \st x(S, j) > 0,\; S \notin \mathcal{S}^2_j, \; \sum_{i \in S \setminus M^+_j} \frac{\mu_i}{\lambda_j} \leq \sum_{i \in S \cap M^+_j} \frac{\mu_i}{\lambda_j}\right\}.$$
  
  \begin{nclaim}
    $\sum_{S \in \mathcal{S}^1_j} x(S, j) \leq \frac{5}{32}$ for all $j \in J_3$.
  \end{nclaim}
  \begin{proof}
  Consider any $j \in J_3$ and any $S \in \mathcal{S}^1_j$. By Lemma~\ref{lem:prices-to-speed} and the fact that for each $i \in S$ we have $\{i\} \in \mathcal{F}_j$, observe that 
    $$\sum_{i \in S \cap M^+_j} g_j(\{i\}) \geq \frac{1}{2}\sum_{i \in S \cap M^+_j} \frac{\mu_i}{\lambda_j} \geq \frac{1}{4}\sum_{i \in S} \frac{\mu_i}{\lambda_j}$$ because $S \in \mathcal{S}^1_j$.
    Note further that $g_j(S) = \frac{1}{2} \left(\sum_{i \in S} \frac{\mu_i}{\lambda_j} + \frac{1}{C}\right) \leq \frac{5}{8}\sum_{i \in S} \frac{\mu_i}{\lambda_j}$ where the first identity follows from $x(S, j) > 0$ and complementary slackness for \eqref{LPgs} and the final inequality follows from $S \notin \mathcal{S}^2_j$, which implies $\sum_{i \in S} \frac{\mu_i}{\lambda_j} > \frac{4}{C}$.
    
    Combining the two inequalities above we conclude that $\sum_{i \in S \cap M^+_j} \frac{g_j(\{i\})}{g_j(S)} \geq \frac{2}{5}$ for all $S \in \mathcal{S}^1_j$.
   This implies 
  \begin{align*}
    \sum_{S \in \mathcal{S}^1_j} x(S, j) \leq \frac{5}{2} \sum_{S \in \mathcal{S}^1_j} \sum_{i \in S \cap M^+_j} \frac{g_j(\{i\})}{g_j(S)} x(S, j) \leq \frac{5}{32},
  \end{align*}
  because $\sum_{i \in M^+_j} \sum_{S \subseteq M} \frac{g_j(\{i\})}{g_j(S)} x(S, j) \leq \frac{1}{16}$ for all $j \in J_3 \subseteq J \setminus J_1$.
\hfill $\blacklozenge$
\end{proof}

\begin{nclaim}
$\sum_{S \in \mathcal{S}^2_j} \left(2 - \frac{1}{Cg_j(S)}\right)x(S,j) \leq \frac{1}{5}$ for all $j \in J_3$.
\end{nclaim}
\begin{proof}
The claim follows directly from $g_j(S) \leq \frac{1}{2}\left(\sum_{i \in S} \frac{\mu_i}{\lambda_j} + \frac{1}{C}\right) \leq \frac{5}{2C}$ for all $S \in \mathcal{S}^2_j$ and the fact that $\sum_{S \in S^2_j} x(S, j) \leq \frac{1}{8}$ for all $j \in J_3 \subseteq J \setminus J_2$. 
\hfill $\blacklozenge$
\end{proof}

Note that every $S \subseteq M$ with $x(S, j) > 0$ is contained in $\mathcal{S}^1_j \cup \mathcal{S}^2_j \cup \mathcal{S}^3_j$.
From feasibility of $x$ and the two claims above we conclude that
\begin{align*}
      \gamma_j = \elsum{S \in \mathcal{S}^3_j} x(S, j)
      & \geq \frac{1}{2} \sum_{S \in \mathcal{S}^3_j} \left(2 - \frac{1}{Cg_j(S)}\right) x(S, j)\\
        & \geq \frac{1}{2} \left(1 - \sum_{S \in \mathcal{S}^1_j \cup \mathcal{S}^2_j} \left(2 - \frac{1}{Cg_j(S)}\right) x(S, j)\right) \\
        & \geq \frac{1}{2} \left(1 - \frac{1}{5} - \frac{5}{16}\right) = \frac{39}{160}
    \end{align*}
  for each $j \in J_3$. This concludes the proof of \cref{lem:split-gamma}. \qed
\end{proof}

We are now ready to prove \cref{lem:gs-split-solution}.

\restategssplit*
\begin{proof}
  Property~\ref{eq:xprime-total} follows immediately from the construction of $x'$, since for every $j \in J_3$:
  $$\sum_{S \subseteq M} x'(S,j) = \frac{\sum_{S \subseteq M} \bar{x}(S, j)}{\gamma_j} = 1.$$
  
For proving that $x'$ fulfills Property~\ref{eq:xprime-matroid}, note that $x'(S, j) > 0$ implies that $S \in  \mathcal{A}_j(T)$ for some $T \subseteq M$ with $x(T, j) > 0$. Thus $S \subseteq T \setminus M_j^+ \in \mathcal{F}_j$ by \cref{lem:matroidstructure}.

To establish Property~\ref{eq:xprime-load}, we will use the following claim.

\begin{nclaim}\label{claim:xprime-makespan}
  $\sum_{j \in J_3} \sum_{S \subseteq M : i \in S} \frac{1}{g_j(S)}x'(S, j) \leq \frac{5C}{2\gamma_j}$ for all $i \in M$.
\end{nclaim}
\begin{proof}
Let $j \in J_3$. Note that 
\begin{align*}
  \elsum{A \in \mathcal{A}_j(T)} g_j(A) \geq \frac{1}{2} \sum_{i \in T \setminus M^+_j} \frac{\mu_i}{\lambda_j} \geq \frac{1}{4} \sum_{i \in T} \frac{\mu_i}{\lambda_j} \geq
  \frac{1}{5} \left(\sum_{i \in T} \frac{\mu_i}{\lambda_j} + \frac{1}{C}\right)
  = \frac{2}{5} g_j(T)
\end{align*}
for all $T \in \mathcal{S}^3_j$, where the first inequality follows from Lemma~\ref{lem:prices-to-speed} and the fact that $\mathcal{A}_j(T)$ is a partition of $T \setminus M^+_j$, while the second and third inequality follow from the definition of $\mathcal{S}^3_j$, which implies $\sum_{i \in T \setminus M^+_j} \frac{\mu_i}{\lambda_j} \geq \frac{1}{2}\sum_{i \in T} \frac{\mu_i}{\lambda_j}$ and $\sum_{i \in T} \frac{\mu_i}{\lambda_j} > \frac{4}{C}$.

 Now consider any machine $i \in M$ and observe that
\begin{align*}
  \sum_{S \subseteq M : i \in S} \frac{1}{g_j(S)}x'(S, j) & = \frac{1}{\gamma_j} \sum_{S \subseteq M : i \in S} \frac{1}{g_j(S)} \sum_{T \in \mathcal{T}_j(S)} \frac{g_j(S)}{\sum_{A \in \mathcal{A}_j(T)} g_j(A)}x(T, j) \\
  & = \frac{1}{\gamma_j} \sum_{T \in \mathcal{S}^3_j : i \in T} \frac{1}{\sum_{A \in \mathcal{A}_j(T)} g_j(A)}x(T, j) \\
  & \leq \frac{5}{2\gamma_j} \sum_{T \in \mathcal{S}^3_j : i \in T} \frac{1}{g_j(T)}x(T, j)
\end{align*}
by construction of $x'$. The claim then follows from the fact that 
$$\sum_{T \in \mathcal{S}^3_j : i \in T} \frac{1}{g_j(T)}x(T, j) \leq \sum_{j \in J_3} \sum_{S \subseteq M : i \in S} \frac{1}{g_j(S)} x(S, j) \leq C$$ by feasibility of $x$. 
\hfill $\blacklozenge$
\end{proof}

Recall that $x'(S, j) > 0$ for $j \in J_3$ and $S \subseteq M$ implies that $S \in \mathcal{A}_j(T)$ for some $T \subseteq M$. In particular, this implies $\sum_{i \in S} \frac{\mu_i}{\lambda_j} \leq \frac{4}{C}$, and hence $$\textstyle g_j(S) \leq \frac{1}{2} \left(\sum_{i \in S} \frac{\mu_i}{\lambda_j} + \frac{1}{C}\right) \leq \frac{5}{2C}$$ for all $S \subseteq M$ and $j \in J_3$ with $x'(S, j) > 0$, by the constraints of \eqref{LPgsD}. Combining this insight with Claim~\ref{claim:xprime-makespan} and the fact that $\gamma_j \geq \frac{39}{160}$ by \cref{lem:split-gamma}, we conclude that
\begin{align*}
  \sum_{j \in J_3} x'_{ij} & = \sum_{j \in J_3} \sum_{S \subseteq M : i \in S} x'(S, j)
  \leq \frac{5}{2C} \sum_{j \in J_3} \sum_{S \subseteq M : i \in S} \frac{1}{g_j(S)} x(S, j)\\
  & \leq \frac{5}{2C} \cdot \frac{5C}{2\gamma_j} \leq \frac{1000}{39} < 26
\end{align*}
for any $i \in M$, thus establishing Property~\ref{eq:xprime-load}.

Finally, to establish Property~\ref{eq:xprime-prices}, we first prove the following claim.

\begin{nclaim}\label{claim:average-prices}
  Let $j \in J_3$ and $T \in \mathcal{S}_j^3$. Then
  $\displaystyle \frac{\sum_{A \in \mathcal{A}_j(T)} g_j(A) \sum_{i \in A} \frac{\mu_i}{\lambda_j}}{\sum_{A \in \mathcal{A}_j(T)} g_j(A)} \geq \frac{79}{40C}$.
\end{nclaim}
\begin{proof}
For $S \subseteq M$ define $\Delta(S) := \sum_{i \in S} \frac{\mu_i}{\lambda_j}$.
  The claim is trivial if $\Delta(A) \geq \frac{2}{C} > \frac{79}{40C}$ for all $\mathcal{A}_j(T)$.
  Thus assume there is $A_0 \in \mathcal{A}_j(T)$ with $\Delta(A_0) < \frac{2}{C}$.
  Note that $\Delta(A) > \frac{4}{C} - \Delta(A_0) > \frac{2}{C}$ for all $A \in \mathcal{A}_j(T) \setminus \{A_0\}$ by construction of $\mathcal{A}_j(T)$. Furthermore, because $T \in \mathcal{S}^3_j$ and, hence, $\Delta(T \setminus M^+_j) \geq \frac{2}{C}$, there is at least one set $A_1 \in \mathcal{A}_j(T) \setminus \{A_0\}$.
  We obtain
  \begin{align*}
    \frac{\sum_{A \in \mathcal{A}_j(T)} g_j(A) \sum_{i \in A} \frac{\mu_i}{\lambda_j}}{\sum_{A \in \mathcal{A}_j(T)} g_j(A)} & \geq \min \left\{ \frac{g_j(A_0) \Delta(A_0) + g_j(A_1) \Delta(A_1)}{g_j(A_0) + g_j(A_1)}, \; \frac{2}{C} \right\},
  \end{align*}
  using the fact that $\frac{a+b}{c+d} \geq \min\{\frac{a}{c}, \frac{b}{d}\}$ for any $c,d > 0$.
  
  Note that
  $\frac{g_j(A_0) \Delta(A_0) + g_j(A_1) \Delta(A_1)}{g_j(A_0) + g_j(A_1)}$
  is bounded from below by the value of the following minimization problem:
  \begin{alignat*}{2}
    \minimize \quad && \frac{g_0 \Delta_0 + g_1 \Delta_1}{g_0 + g_1}\\
    \subjectto \quad && \Delta_0 + \Delta_1 & \geq \frac{4}{C}\\
    && \frac{1}{2} \Delta_0 \leq g_0  & \leq \frac{1}{2} (\Delta_0 + \frac{1}{C})\\
    && \frac{1}{2} \Delta_1 \leq g_1  & \leq \frac{1}{2} (\Delta_1 + \frac{1}{C})\\
    && g_0, g_1, \Delta_0, \Delta_1 & \geq 0
  \end{alignat*}
  By applying KKT conditions, it is easy to see that the minimum is attained when $\Delta_0 + \Delta_1 = \frac{4}{C}$, $g_0 = \frac{1}{2}(\Delta_0 + \frac{1}{C})$, and $g_1 = \frac{1}{2}\Delta_1$.
  From this, by setting $\Delta_0 = \frac{7}{4C}$, $\Delta_1 = \frac{9}{4C}$, $g_0 = \frac{11}{8C}$ and $g_1 = \frac{9}{8C}$, it is easy to derive that the optimal value of the above optimization problem is $\frac{79}{40C}$. 
\hfill$\blacklozenge$
\end{proof}

Claim~\ref{claim:average-prices}, together with the fact that $\gamma_j = \sum_{T \in \mathcal{S}^3_j} x(T, j)$, immediately implies that 
\begin{align*}
  \sum_{i \in M} \frac{\mu_i}{\lambda_j} x'_{ij} & = \sum_{S \subseteq M} x'(S, j) \sum_{i \in S} \frac{\mu_i}{\lambda_j} \\ 
  & = \sum_{S \subseteq M} \sum_{T \in \mathcal{T}_j(S)} \frac{g_j(S)}{\sum_{A \in \mathcal{A}_j(T)} g_j(A)} \frac{x(T,j)}{\gamma_j} \sum_{i \in S} \frac{\mu_i}{\lambda_j} \\ 
  &= \sum_{T \in \mathcal{S}^3_j} \frac{x(T, j)}{\gamma_j}  \frac{\sum_{A \in \mathcal{A}_j(T)} g_j(A) \sum_{i \in A} \frac{\mu_i}{\lambda_j}}{\sum_{A \in \mathcal{A}_j(T)} g_j(T)} \geq \frac{79}{40C},
\end{align*}
which concludes the proof of \cref{lem:gs-split-solution}.
\qed
\end{proof}

\subsection{\cref{sec:step3B}: Proofs of \cref{lem:pm-LP-gs,lem:gs-pmi-value}}

\restatePMIntegrality*
\begin{proof}
Note that the function $r'_{jk}(U) := \min \{r_j(U \cap M_{jk}), d_{jk}\}$ results from restricting the rank function $r_j$ of the matroid $\mathcal{F}_j$ to the set $M_{jk}$ and then truncating it at $d_{jk}$. As these operations preserve submodularity, $r'_j = \sum_{k=0}^{\infty} r'_{jk}$ is a submodular function. Thus the constraints $\sum_{i \in U} y_{ij} \leq r'_j(U)$ for all $U \subseteq M$ and $j \in J_3$ describe a polymatroid (corresponding to the direct sum of the polymatroids defined by each function $r'_j$ for $j \in J_3$, respectively).
  We conclude that \eqref{LPpm} corresponds to the intersection of two polymatroids and therefore has only integral extreme points and optimal solutions to the LP can be computed in polynomial time~\cite[Theorem 46.1]{schrijver03}. \qed 
\end{proof}

\restatePolymatroidValue*
\begin{proof}
For $j \in J_3$ and $k \in \mathbb{Z}$ define $\gamma_{jk} := d_{jk} / \sum_{i \in M_{jk}} x'_{ij}$ if $d_{jk} > 0$ and $\gamma_{jk} = 0$, otherwise.
  For $i \in M_{jk}$ define $y'_{ij} := \gamma_{jk} x'_{ij}$. We show that $y'$ is a feasible solution to \eqref{LPpm}.
  First, let $i \in M$. Using the fact that $\gamma_{jk} \leq 1$, by construction, and that $\sum_{j \in J_3} x'_{ij} \leq 26$, by Property \eqref{eq:xprime-load} of Lemma \ref{lem:gs-split-solution}, we obtain $\sum_{j \in J_3} y'_{ij} \leq 26$.
  Now let $j \in J_3$ and $U \subseteq M$. For $k \in \mathbb{Z}$ we obtain
  \begin{align*}
    \elsum{i \in U \cap M_{jk}} y'_{ij} &
    \,=\, \frac{d_{jk} \cdot \sum_{i \in U\cap M_{jk}} x'_{ij}}{\sum_{i \in M_{jk}} x'_{ij}}
    \,\leq\, \min \left\{d_{jk},\ \elsum{i \in U \cap M_{jk}} x'_{ij}\right\}
  \end{align*}
    by definition of $y'$ and $d_{jk}$. We further observe that
  \begin{align*}
    \elsum{i \in U \cap M_{jk}} x'_{ij} & 
     = \sum_{S \subseteq M} |S \cap U \cap M_{jk}| \, \cdot \, x'(S, j) \leq r_j(U \cap M_{jk}),
  \end{align*}
  where the inequality follows from $\sum_{S \subseteq M} x'(S, j) \leq 1$ for all $j \in J_3$ and 
  the fact that $x'(S, j) > 0$ implies $S \in \mathcal{F}_j$, by Property \ref{eq:xprime-matroid} of Lemma \ref{lem:gs-split-solution}, and hence $|S \cap U \cap M_{jk}| \leq r_j(U \cap M_{jk})$.
  We thus get $$\elsum{i \in U \cap M_{jk}} y'_{ij} \leq \min \{d_{jk},\, r(U \cap M_{jk})\}.$$
  As $M_{jk} \cap M_{jk'} = \emptyset$ for $k \neq k'$, we conclude that $\sum_{i \in U} y'_{ij} \leq r'(U)$. Hence $y'$ is a feasible solution to \eqref{LPpm}. 
  Finally, note that $\sum_{i \in M} \sum_{j \in J_3} y'_{ij} = \sum_{j \in J_3} \sum_{k = 0}^{\infty} d_{jk}$ by construction of $y'$. 
\qed
\end{proof}

\section{Appendix for \cref{sec:fair-allocation}: Proofs of \cref{cor:santa,cor:mmfa-natural}}

\restateSantaClaus*
\begin{proof}
Consider an instance of \textsc{Santa Claus} with agent set $A$, item set $I$, and utilities of the form $u_j(S) = \sum_{i \in S} u_{ij}$, with $u_{ij} \geq 0$ for each $i \in I$ and $j \in A$.
As described in the proof of \cref{thm:mmfa}, given a target value $V$, we can construct an instance of generalized malleable scheduling with $J := A$ and $M := I$ and processing speed functions $g_j(S) = \min \{u_j(S), V\}$ for $j \in J$.
We observe $f_j(S) = 1/g_j(S)$ fulfills the properties of a speed-implementable processing time function given in~\citep{fotakis2020malleable} with speeds $s_{ij} := u_{ij}$.

The rounding scheme described in~\citep[Section~3.1]{fotakis2020malleable} constructs an assignment $\mathbf{T}$ from an LP relaxation~\citep[Section~2]{fotakis2020malleable} with the property that $i \in T_j$ implies $x_{ij} > 0$ (where the $x_{ij}$ are the variables of the LP).
Importantly, the solutions to the LP fulfill
\begin{align}
	\sum_{j \in J : i \in T_j} f_{j}(\{i\}) x_{ij} \leq C \label{eq:speed-lp}
\end{align} 
(this follows from constraint (2) of the LP presented in~\citep{fotakis2020malleable}, where some of the coefficients $f_{j}(\{i\})$ are actually scaled up by a factor $\gamma_j/s_{ij} \geq 1$).

In the construction of the assignment $\mathbf{T}$, the jobs are split in two sets $\mathcal{J}^{(1)}$ and $\mathcal{J}^{(2)}$ such that each machine 
\begin{itemize}
	\item each machine processes at most one job from $\mathcal{J}^{(2)}$ and each $j \in \mathcal{J}^{(2)}$ is assigned to a set of machines with $f_j(T_j) \leq 2C$~\citep[Proposition~9]{fotakis2020malleable},
	\item and each  job $j \in \mathcal{J}^{(1)}$ is assigned to a single machine $i$ with $x_{ij} \geq \frac{1}{2}$.
\end{itemize}
Moreover 
$\sum_{j \in \mathcal{J}^{(1)} : i \in T_j} f_{j}(\{i\}) \leq 2 \cdot \sum_{j \in \mathcal{J}^{(1)} : i \in T_j} f_{j}(\{i\}) x_{ij}$ for each machine $i \in M$~\citep[Proposition~8]{fotakis2020malleable}. 
As a consequence, $f_{j}(\{i\}) \leq 2C$ and no machine is  assigned more than two jobs from $\mathcal{J}^{(1)}$. 
Moreover, if a machine $i$ is assigned two jobs from $j_1, j_2 \in \mathcal{J}^{(1)}$, then $x_{ij_1}, x_{ij_2} \geq \frac{1}{2}$ together with \eqref{eq:speed-lp} implies $x_{ij} = 0$ for all other jobs. 
Thus a machine can be assigned at most two jobs (either two jobs from $\mathcal{J}^{(1)}$, or one jobs from $\mathcal{J}^{(1)}$ and one from $\mathcal{J}^{(2)}$) and $g_j(T_j) = 1/f_j(T_j) \geq \frac{1}{2C}$ for all $j \in J$.
Translated back to the instance of \textsc{Santa Claus}, we conclude that in $\mathbf{T}$ no item is assigned more than two time and every agent receives a utility of $u_j(T_j) \geq g_j(T_j) \geq \frac{1}{2C} = \frac{V}{2}$. \qed
\end{proof}

\restateMMFAnatural*
\begin{proof}
Consider an instance of MMFA with agent set $A$, item set $I$, and $M^\natural$-concave utilities $u_j$ for each $j \in A$.
Let $V$ be a target value for the minimum utility.
For $i \in M$ and $j \in J$ define $p_{ij} := \max \{g_{j}(\{i\}) - V, 0\}$.
Construct an instance of generalized malleable scheduling with $J := A$ and $M := I$ and processing speed functions $g_j(S) := u_j(S) - \sum_{i \in S} p_{ij}$.
Note that $g_j$ inherits $M^\natural$-concavity from $u_j$.
Apply the constant-factor approximation for the \textsc{Assignment} problem from \cref{thm:assignment-approx} to this instance with load bound $C := 1/V$.

We first argue that if the algorithm concludes that no assignment of load at most $C$ exists, there is no solution to the MMFA instance that secures every agent a utility of at least $V$.
Assume there is a solution $\mathbf{T}$ to the MMFA instance with minimum utility at least $V$, then the following assignment $\mathbf{S}$ has a load of at most $C$ in the generalized malleable scheduling instance: For $j \in J$, let $S_j := T_j$ if $u_{j}(\{i\}) \leq V$ for all $i \in T_j$; otherwise, let $S_j := \{i\}$ for some $i \in T_j$ with $u_{j}(\{i\}) > V$.
Note that in the assignment $\mathbf{S}$ each machine is assigned at most one job (because in $\mathbf{T}$ every item is assigned to at most one agent) and that $g_{j}(S_j) \geq V$ by construction. Hence $L(\mathbf{S}) \leq 1/V = C$.

Now assume that the algorithm returns an assignment $\mathbf{S}$ with $L(\mathbf{S}) \leq 193C$. We show that the same assignment $\mathbf{S}$ provides a utility of $\frac{1}{193}V$ to every agents while assigning every item only to a constant number of agents.
Note that $u_j(S_j) \geq g_j(S_j) \geq \frac{1}{193C} = \frac{1}{193}V$ for all $j \in A$, where the first inequality follows from construction of $g_j$ and the second follows from $L(\mathbf{S}) \leq 193C$.
Moreover, observe that in the assignment $\mathbf{S}$ constructed by the algorithm, every machine is assigned to at most $20$ jobs from $J_2$ (see \cref{sec:step2}) and at most $26$ jobs from $J_3$ (see \cref{sec:step3A,sec:step3B}).
In addition, note that every job $j \in J_1$ is assigned to exactly one machine and that $\sum_{j \in J_1 : i \in S_j} \frac{1}{g_j(\{i\})} \leq 32C$ for every machine $i \in M$ (see \cref{sec:step1}). Because $g_j(\{i\}) = u_j(\{i\}) - p_{ij} \leq V = 1/C$ by construction of $g_j$, this implies that every machine is assigned to at most $32$ jobs from $J_1$. We conclude that $\mathbf{S}$, interpreted as a solution to the MMFA instance, assigns every item to at most $78$ agents, and every agent receives utility at least $\frac{1}{193}V$. \qed
\end{proof}